\documentclass[11pt]{article}
\usepackage{framed}
\usepackage{amssymb,amsfonts,amsmath}
\usepackage{algorithm2e,framed}
\usepackage{amsfonts,amsmath}
\usepackage{graphicx}
\usepackage{subfigure,enumerate,bm,amsmath,amsthm,wrapfig,array,
color,algorithmic}
\usepackage{graphics}
\usepackage{multirow}
\usepackage{setspace}

\def\eE{\e}

\long\def\symbolfootnote[#1]#2{\begingroup%
\def\thefootnote{\fnsymbol{footnote}}\footnote[#1]{#2}\endgroup}

\newcommand{\FNorm }[1]{\mbox{}\|#1\|_\mathrm{F}  }
\newcommand{\FNormS}[1]{\mbox{}\|#1\|_\mathrm{F}^2}

\newcommand{\TNorm }[1]{\mbox{}\|#1\|_2  }

\newcommand{\XNorm }[1]{\mbox{}\|#1\|_{\xi}  }

\newtheorem{theorem}{\bf Theorem}[]
\newtheorem{lemma}[theorem]{Lemma}

\newtheorem{corollary}[theorem]{Corollary}

\newcommand{\qedsymb}{\hfill{\rule{2mm}{2mm}}}

\newcommand{\transp}{^{\textsc{T}}}

\newcommand{\mat}[1]{{\ensuremath{\bm{\mathrm{#1}}}}}

\newcommand{\pinv}[1]{ {#1}^\dagger}

\def\rank{\hbox{\rm rank}}

\def\b{{\mathbf b}}
\def\e{{\mathbf e}}

\def\v{{\mathbf v}}

\def\matA{\mat{A}}
\def\matB{\mat{B}}
\def\matC{\mat{C}}

\def\matE{\mat{E}}

\def\matH{\mat{H}}
\def\matI{\mat{I}}

\def\matS{\mat{S}}

\def\matU{\mat{U}}
\def\matV{\mat{V}}

\def\matX{\mat{X}}
\def\matY{\mat{Y}}
\def\matZ{\mat{Z}}
\def\matSig{\mat{\Sigma}}

\def\matOmega{\mat{\Omega}}

\def\scl{{\textsc{l}}}

\DeclareMathSymbol{\Prob}{\mathbin}{AMSb}{"50}
\newcommand\remove[1]{}

\def\math#1{$#1$}

\def\frac#1#2{{#1\over #2}}

\DeclareMathSymbol{\R}{\mathbin}{AMSb}{"52}

\def\x{{\mathbf x}}

\def\z{{\mathbf z}}

\def\b{{\mathbf b}}

\def\norm#1{{\|#1\|}}

\def\ceil#1{{\left\lceil\,#1\,\right\rceil}}

\def\dotfil{\leaders\hbox to 1.5mm{.}\hfill}
\newcounter{rmnum}
\def\RN#1{\setcounter{rmnum}{#1}\uppercase\expandafter{\romannumeral\value{rmnum}}}
\def\rn#1{\setcounter{rmnum}{#1}\expandafter{\romannumeral\value{rmnum}}}

\begin{document}

\title{\bf A Note on Sparse Least-squares Regression
}

\author{
{\bf Christos Boutsidis} \\
Mathematical Sciences Department \\
IBM T.J. Watson Research Center \\
cboutsi@us.ibm.com
\and
{\bf Malik Magdon-Ismail} \\
Computer Science Department \\
Rensellaer Polytechnic Institute \\
magdon@cs.rpi.edu
}

\maketitle

\begin{abstract}
We compute a \emph{sparse} solution 
to the classical least-squares problem 
$\min_\x\TNorm{\matA \x -\b},$ where
 $\matA$ is an arbitrary matrix.
We describe a novel algorithm for this sparse least-squares problem. The algorithm operates as follows: first, it selects columns from $\matA$, and then solves a least-squares problem  only with the selected columns. The column selection algorithm that we use is known to perform well for the well studied column subset selection problem. The contribution of this article is to show that it gives favorable results for sparse least-squares as well.
Specifically, we prove that the solution vector obtained by our algorithm is close to the solution vector obtained via what is known as the `` SVD-truncated regularization approach''.
\end{abstract}

\section{Introduction}
Fix inputs $\matA \in \R^{m \times n}$ and $\b \in \R^m$. We study 
least-squares regression:
$ \min_{\x \in \R^n} \TNorm{\matA \x - \b}.$
It is well known that the minimum norm solution vector
can be found
using the pseudo-inverse of $\matA$: 
$\x^{*} = \pinv{\matA}\b  = ( \matA\transp \matA)^{-1} \matA\transp \b.$
When \math{\matA} is ill-conditioned, \math{\pinv{\matA}} becomes unstable to perturbations and overfitting
can become a serious problem. For example, when the smallest non-zero singular value of $\matA$ is close to zero,
the largest singular value of \math{\pinv{\matA}} can be extremely large and the solution vector $\x^{*} = \pinv{\matA}\b$ obtained via a numerical algorithm is not
the optimal, due to numerical instability issues. Practitioners deal with such situations using \emph{regularization}.

Popular regularization techniques are the
Lasso~\cite{Tib96}, the Tikhonov regularization~\cite{GHO00}, and the
truncated SVD~\cite{Han87}. The lasso minimizes
$\TNorm{\matA \x - \b}  + \lambda ||\x||_1,$
and Tikhonov regularization minimizes
$\TNorm{\matA \x - \b}^2  + \lambda ||\x||_2^2$
(in both cases $\lambda > 0$ is the regularization parameter).
The truncated SVD minimizes
$\TNorm{\matA_k \x - \b} ,$
where $k < \rank(\matA)$ is a rank parameter and
 $\matA_k \in \R^{m \times n}$ 
is
the best rank-$k$ approximation to $\matA$ obtained
via the SVD. 
So,
the truncated SVD 
solution is
$
\x_k^*= \pinv{\matA}_k \b.
$
Notice that these regularization
methods impose parsimony on $\x$ in different ways. A
combinatorial approach to 
regularization is to explicitly impose the sparsity constraint on~$\x$, 
requiring it to have few non-zero elements.
 We give a new deterministic algorithm which, 
for $r=O(k)$,
computes an $\hat{\x}_r \in \R^{n}$ with at most $r$ non-zero entries such that
$\TNorm{\matA \hat{\x}_r - \b} \approx \TNorm{\matA \x_k^* - \b}.$

\subsection{Preliminaries}
The compact (or thin) Singular Value Decomposition (SVD) of a matrix $\matA \in \R^{m \times n}$ of rank $\rho$ is
\begin{eqnarray*}
\label{svdA} \matA
         = \underbrace{\left(\begin{array}{cc}
             \matU_{k} & \matU_{\rho-k}
          \end{array}
    \right)}_{\matU_\matA \in \R^{m \times \rho}}
    \underbrace{\left(\begin{array}{cc}
             \matSig
_{k} & \bf{0}\\
             \bf{0} & \matSig
_{\rho - k}
          \end{array}
    \right)}_{\matSig_\matA \in \R^{\rho \times \rho}}
    \underbrace{\left(\begin{array}{c}
             \matV_{k}\transp\\
             \matV_{\rho-k}\transp
          \end{array}
    \right)}_{\matV_\matA\transp \in \R^{\rho \times n}},
\end{eqnarray*}
Here, $\matU_k \in \R^{m \times k}$ and $\matU_{\rho-k} \in \R^{m \times (\rho-k)}$ contain the left singular vectors of $\matA$.
Similarly, $\matV_k \in \R^{n \times k}$ and $\matV_{\rho-k} \in \R^{n \times (\rho-k)}$
contain the right singular vectors. The singular values of $\matA$, which we denote as
\math{\sigma_1(\matA)\ge\sigma_2(\matA)\ge\cdots\ge\sigma_\rho(\matA)>0} are
contained in \math{\matSig
_k \in \R^{k \times k}} and \math{\matSig
_{\rho-k}
\in \R^{(\rho-k) \times (\rho-k)}}.
We use $\pinv{\matA} = \matV_\matA \matSig
_\matA^{-1} \matU_\matA\transp \in \R^{n \times m}$
to denote the Moore-Penrose
pseudo-inverse of $\matA$ with $\matSig
_\matA^{-1}$ denoting the inverse of
$\matSig
_\matA$. Let $\matA_k=\matU_k \matSig
_k \matV_k\transp \in \R^{m \times n}$ and
$\matA_{\rho-k} = \matA - \matA_k = \matU_{\rho-k}\matSig
_{\rho-k}\matV_{\rho-k}\transp \in \R^{m \times n}$.
\begin{algorithm}[!h]
\begin{framed}
\caption{
{\sf Deterministic Sparse Regression}
}
\label{alg:deterministic}
\begin{algorithmic}[1]
\STATE 
{\bf Input:} \math{\matA \in \R^{m \times n}}, $\b \in \R^m,$
target rank \math{k < \rank(\matA)}, and parameter \math{0 < \varepsilon < 1/2}.
\STATE Obtain \math{\matV_k \in \R^{n \times k}}
from the SVD of $\matA$ and compute \math{\matE =\matA-\matA\matV_k\matV_k\transp \in \R^{m \times n}}.
\STATE Set
$\matC = \matA\matOmega\matS \in \R^{m \times r}$,
with $r = \ceil{\frac{9k}{\varepsilon^2}}$ and
\\
\hspace*{-0.05in}\math{[\matOmega, \matS]={\sf DeterministicSampling}(\matV_k\transp,\matE,r),}
\STATE Set $\x_r = \pinv{\matC}\b \in \R^{r}$, 
and $\hat{\x}_r=\matOmega\matS\x_r \in \R^{n}$ ($\hat\x_r$ has at most
\math{r} non-zeros 
at the 
indices of the selected columns in $\matC$).
\STATE {\bf Return} $\hat{\x}_r \in \R^{n}$.
\end{algorithmic}
\end{framed}
\end{algorithm}
For $k < \rank(\matA)$, the SVD gives the best rank $k$ approximation to \math{\matA} in both the
spectral and the Frobenius norm: for $\tilde\matA \in \R^{m \times n}$, let \math{\rank(\tilde\matA) \le k}; then, for $\xi=2,\mathrm{F}$,
\math{\XNorm{\matA-\matA_k}\le\XNorm{\matA-\tilde\matA}}. Also,
$\TNorm{\matA-\matA_k} = \TNorm{\matSig
_{\rho-k}} = \sigma_{k+1}(\matA)$, and
$\FNormS{\matA-\matA_k} = \FNormS{\matSig
_{\rho-k}} = \sum_{i=k+1}^{\rho}\sigma_{i}^2(\matA) $.
The Frobenius and the spectral norm of $\matA$ are defined as: $ \FNormS{\matA} =
\sum_{i,j} \matA_{ij}^2 =
\sum_{i=1}^\rho\sigma_i^2(\matA)$;
and $\TNorm{\matA} = \sigma_1(\matA)$.
Let $\matX$ and $\matY$ be matrices
of appropriate dimensions; then,
$\FNorm{\matX\matY} \leq \min\{\FNorm{\matX}
\TNorm{\matY},\TNorm{\matX}
\FNorm{\matY}\}$. This is a stronger version of the standard
submultiplicativity property $\norm{\matX\matY}_{\mathrm{F}} \leq \norm{\matX}_{\mathrm{F}}
\norm{\matY}_{\mathrm{F}}$, which we will refer to as ``spectral
submultiplicativity''.

Given  $k < \rho=\rank(\matA)$, the truncated rank-\math{k}  SVD regularized
weights are
$$ \x_{k}^* = \pinv{\matA}_k \b = \matV_k\matSig
_k^{-1} \matU_k\transp \b \in \R^{n},$$
and note that
$  \TNorm{ \b -  \matA_k \pinv{\matA}_k\b}=\TNorm{ \b -  \matU_k \matU_k\transp \b}$

Finally, for $r < n$, 
let \math{\matOmega=[\z_{i_1},\ldots,\z_{i_r}] \in \R^{n \times r}}
where \math{\z_i \in \R^{m}} are standard basis vectors;
\math{\matOmega} is a \emph{sampling matrix} because
\math{\matA\matOmega\in \R^{m \times r}}  is a matrix whose columns are 
sampled (with possible repetition) from the columns of~\math{\matA}.
Let $\matS \in \R^{r \times r}$ 
be a diagonal \emph{rescaling matrix} with positive entries; then,
we define the sampled and rescaled columns from \math{\matA} by
$\matC = \matA \matOmega\matS$: \math{\matOmega} samples some columns from $\matA$ and
then \math{\matS} rescales them.

\section{Results} \label{sec0}
Our sparse solver to minimize
\math{\norm{\matA\x-\b}_2}
takes as input the sparsity parameter $r$ 
(i.e., the solution vector $\x$ is allowed at most $r$ non-zero entries), 
and 
selects $r$ rescaled 
columns from $\matA$ (denoted by  $\matC$). We
then solve the least-squares problem to minimize
$||\matC \x - \b||_2$.
The result is a dense vector $\pinv{\matC} \b$ with $r$ dimensions. 
The sparse solution \math{\hat\x_r} will be zero at indices 
 corresponding to columns not selected in $\matC$, and we use 
$\pinv{\matC} \b$ to compute the other entries of \math{\hat\x_r}.
\label{sec:detalgo}
\begin{algorithm}[t]
\begin{framed}
   \caption{{\sf DeterministicSampling} (from \cite{BDM11a})}
\begin{algorithmic}[1]
   \STATE {\bfseries Input:} $\matV\transp=[\v_1,\ldots,\v_n] \in \R^{k \times n}$;
\math{\matE=[\eE_1,\ldots,\eE_n] \in \R^{m \times n}}; and \math{r > k}.
   \STATE {\bfseries Output:} Sampling and rescaling matrices \math{\matOmega
 \in \R^{n \times r},\matS \in \R^{r \times r}}.
   \STATE Initialize  $\matB_0 = \bm{0}_{k \times k}$, $\matOmega
 = \bm{0}_{n \times r}$, $\matS=\bm{0}_{r \times r}$.
   \FOR{$\tau=0$ {\bfseries to} $r-1$}
   \STATE Set \math{\scl_\tau=\tau-\sqrt{rk}}.
   \STATE Pick index $i\in\{1,2,...,n\}$ and $t$ such that
    $U(\eE_i)\le\frac{1}{t}\le
       L(\v_i,\matB_{\tau},\scl_\tau).$
   \STATE Update $\matB_{\tau+1}= \matB_{\tau} + t \v_i\v_i\transp$. Set $\matOmega
_{i,\tau+1} = 1$ and $\matS_{\tau+1,\tau+1} = 1/\sqrt{t}$.
   \ENDFOR
   \STATE {\bfseries Return:} \math{\matOmega
 \in \R^{n \times r},\matS \in \R^{r \times r}}.
\end{algorithmic}
\end{framed}
\end{algorithm}
\begin{theorem} \label{theorem2}
Let $\matA \in \R^{m \times n}$, $\b \in \R^{m}$,
rank $k < \rank(\matA) $, and $0 < \varepsilon < 1/2$.
Algorithm~\ref{alg:deterministic} runs in time $O(m n \min\{m,n\} + n k^3 / \varepsilon^2)$ and
returns $\hat\x_r \in \R^n$ with at most
$r = \ceil{9 k / \varepsilon^2}$ non-zero entries such that:
$$
\TNorm{\matA \hat\x_r-\b}
\le
\TNorm{ \matA \x_k^*-\b} + (1+\varepsilon) \cdot \norm{\b}_2 \cdot 
\frac{\norm{\matA-\matA_k}_{\mathrm{F}}}{\sigma_{k}(\matA)}.
$$
\end{theorem}
This upper bound is ``small'' when $\matA$ is ``effectively'' 
low-rank, i.e., $\FNorm{\matA-\matA_k} / \sigma_k(\matA) \ll 1$.
Also, a trivial bound is 
$\TNorm{\matA \hat\x_r-\b}\le\TNorm{\b}$ 
(error when $\hat\x_r$ is the all-zeros vector),
because
$\TNorm{ \matC \pinv{\matC}\b - \b} \le \TNorm{\matC {\bf 0}_{r \times 1} - \b}
=  \TNorm{\b}$.

In the heart of Algorithm~\ref{alg:deterministic} lies a method for selecting columns from $\matA$ (Algorithm 2), which was originally developed in~\cite{BDM11a} for column subset selection, where one selects columns $\matC$ from $\matA$ to  minimize $\FNorm{ \matA - \matC \pinv{\matC} \matA }$. Here, we adopt the same algorithm for least-squares.

The main tool used to prove Theorem~\ref{theorem2}
is a new ``structural'' result that may be of independent interest.
\begin{lemma} \label{result2}
Fix $\matA \in \R^{m \times n}$, $\b \in \R^{n}$, rank $k < \rank(\matA)$, and sparsity  $r > k$.
Let $\x_k^* = \pinv{\matA}_k \b \in \R^n$, where $\matA_k \in \R^{m \times n}$ is the rank-$k$ SVD approximation to $\matA$.
Let $\matOmega
\in \R^{n \times r}$ and $\matS \in \R^{r \times r}$ be any sampling and
rescaling matrices with $\rank(\matV_k\transp\matOmega \matS) = k$.
Let $\matC = \matA\matOmega
\matS \in \R^{m \times r}$ be a matrix of sampled rescaled columns of 
\math{\matA}
and 
let $\hat\x_r =\matOmega
\matS\pinv{\matC}\b\in \R^{n}$ (having at most $r$ non-zeros).
Then,
$$ \TNorm{\matA \hat\x_r - \b} \le
  \TNorm{\matA \x_k^* - \b} +
\TNorm{(\matA - \matA_k) \matOmega
\matS \pinv{ \left( \matV_k\transp\matOmega
\matS \right)}\matSig
_k\matU_k\transp\b}.
$$
\end{lemma}
The lemma says that if the sampling matrix satisfies a simple
rank condition, then solving the regression on the sampled columns
gives a sparse solution to the original problem with a performance
guarantee.

\subsection{Algorithm Description} \label{sec6}
Algorithm~\ref{alg:deterministic} selects \math{r} columns from \math{\matA} to
form \math{\matC} and the corresponding sparse vector $\hat\x_r$.
The core of Algorithm~\ref{alg:deterministic} is the subroutine
{\sf DeterministicSampling}, which is a method to simultaneously
sample the columns of two matrices, while controlling their spectral and
Frobenius norms.
{\sf DeterministicSampling} takes inputs \math{\matV\transp\in
\R^{k\times n}}
and \math{\matE\in\R^{m\times n}}; the matrix
\math{\matV} is orthonormal, \math{\matV\transp\matV=\matI_k}.
(In our application, \math{\matV\transp=\matV_k\transp} 
and $\matE = \matA - \matA_k$.)
We view \math{\matV\transp} and \math{\matE} 
as two sets of \math{n} column vectors,
$\matV\transp=[\v_1,\ldots,\v_n],$ and
$\matE=[\eE_1,\ldots,\eE_n].$

Given $k$ and $r$ and the iterator $\tau = 0, 1,2,...,r-1,$ define 
$\scl_\tau=\tau-\sqrt{rk}$.
For a symmetric
matrix \math{\matB\in\R^{k\times k}} with eigenvalues \math{\lambda_1,\ldots,
\lambda_k} and $\scl\in\R$, define functions
$
\phi(\scl, \matB) =  \sum_{i=1}^k\frac{1}{\lambda_i-\scl},
$
and
$ L(\v, \matB, \scl)  =
\frac{\v\transp(\matB-\scl'\matI_k)^{-2}\v}
{\phi(\scl', \matB)-\phi(\scl,\matB)} -\v\transp(\matB-\scl'\matI_k)^{-1}\v,$
where \math{\scl'=\scl+1}. Also, for a column $\eE$, define
$U(\eE) = \frac{\eE\transp\eE}{\norm{\matA}_{\mathrm{F}}
^2}\left(1-\sqrt{k/r}\right).$
At step \math{\tau}, the algorithm selects any column
\math{i} for which
$U(\eE_i)\le L(\v_i,\matB,\scl_\tau)$ and computes a 
weight $t$ 
such that $U(\eE_i)\le t^{-1} \le L(\v_i,\matB,\scl_\tau)$;
Any $t^{-1}$ in the interval is acceptable.
(There is always at least one such index $i$ (see Lemma 8.1 in~\cite{BDM11a}).)

The running time is dominated by the search for a column which
satisfies \math{U\le L}. To compute \math{L}, one needs
\math{\phi(\scl,\matB)}, and hence the eigenvalues of
\math{\matB}, and \math{(\matB-\scl'\matI_k)^{-1}}. This takes
\math{O(k^3)} time once per iteration, for a total of
\math{O(rk^3)}. Then, for \math{i=1,\ldots,n}, we need to compute
\math{L} for every \math{\v_i}. This takes \math{O(nk^2)} per iteration,
for a total of \math{O(nrk^2)}. To compute \math{U}, we need
\math{\eE_i\transp\eE_i} for \math{i=1,\ldots,n} which takes
\math{O(mn)}. So, in total, {\sf DeterministicSampling} takes \math{O(nrk^2+mn)} time, hence
Algorithm~\ref{alg:deterministic} needs \math{O( m n \min\{m,n\} + n k^3 / \varepsilon^2 )} time.

{\sf DeterministicSampling} uses a greedy procedure to 
sample columns of \math{\matV_k^T} that satisfy the next 
Lemma.
\begin{lemma}[\cite{BDM11a}]
\label{theorem:2setGeneralF}
On  \math{\matV\transp\in\R^{k\times n}}, \math{\matE\in\R^{m\times n}}, and $r > k$
{\sf DeterministicSampling} returns $\matOmega
, \matS$ satisfying

$\sigma_k(\matV\transp \matOmega
 \matS) \ge 1 - \sqrt{k/r}, \qquad \norm{\matE\matOmega
 \matS}_{\mathrm{F}}
  \le  \norm{\matE}_{\mathrm{F}}.$
\end{lemma}
By Lemma~\ref{theorem:2setGeneralF},
Algorithm~\ref{alg:deterministic}returns
$\matOmega, \matS$ that
satisfy the rank condition in Lemma~\ref{result2}, so the 
structural bound applies.
Lemma~\ref{theorem:2setGeneralF}
also bounds two key terms in the bound which ultimately
allow us to prove Theorem~\ref{theorem2}.

\subsection{Proofs} \label{sec4}
\paragraph{Proof of Theorem \ref{theorem2}}
\label{sec:proofdet}
By Lemma~\ref{theorem:2setGeneralF}, $\rank(\matV_k\transp\matOmega
\matS)=k$
so the bound in 
Lemma~\ref{result2} holds.
Recall 
\math{\matE=\matA-\matA_k}.
By submultiplicativity,
\math{\TNorm{\matE \matOmega
 \matS (\matV_k\transp\matOmega
 \matS)^+
\matSig
_k^{-1}\matU_k\transp\b}}
is at most
$$\TNorm{\matE \matOmega
 \matS}
\TNorm{\pinv{(\matV_k\transp\matOmega\matS)}}
\TNorm{\matSig
_k^{-1}\matU_k\transp\b}.
$$
We now bound each term to obtain Theorem~\ref{theorem2}:
\begin{align*}
&\TNorm{\matE \matOmega
 \matS}\le\FNorm{\matE \matOmega
 \matS}\le\FNorm{\matE}=\FNorm{\matA - \matA_k}\tag{a}
\\
&\TNorm{\pinv{(\matV_k\transp\matOmega\matS)}}
=
\frac{1}{\sigma_k(\matV_k\transp\matOmega\matS)}
\le
\frac{1}{1 - \sqrt{k/r}}\le1+\varepsilon \tag{b}
\\
&\TNorm{\matSig
_k^{-1}\matU_k\transp\b}
\le
\TNorm{\matSig
_k^{-1}}
\TNorm{\matU_k\transp}\TNorm{\b}= \TNorm{\b} / \sigma_k(\matA)
\tag{c}\\
\end{align*}
\math{(a)} follows from
Lemma~\ref{theorem:2setGeneralF};
\math{(b)} also follows from
Lemma~\ref{theorem:2setGeneralF} using $r=\lceil{9k/\varepsilon^2}\rceil$
and \math{\varepsilon<1/2}; \math{(c)} follows from 
submultiplicativity.
\qedsymb
\paragraph{Proof of Lemma~\ref{result2}}
We will prove a more general result, and 
Lemma~\ref{result2} will be a simple corollary.
We first introduce a general matrix approximation problem and
present an algorithm for this problem (Lemma \ref{result1}).
Lemma \ref{result2} is a corollary of Lemma \ref{result1}.

Let \math{\matB\in\R^{m\times\omega}} be a matrix which we would like to
approximate; let \math{\matA\in\R^{m\times n} }
be the matrix which we will use to
approximate \math{\matB}. Specifically, we want a \emph{sparse} approximation
of \math{\matB} from \math{\matA}, which means that we would like to
choose \math{\matC\in\R^{m\times r}}
 consisting of \math{r< n} columns
from \math{\matA} such that
\math{\norm{\matB-\matC\pinv{\matC}\matB}_{\mathrm{F}}} is small.
If \math{\matA=\matB}
, then, this is the column based matrix approximation
problem, which has received much interest
recently~\cite{BMD09,BDM11a}.
The more general problem which we study here, with
\math{\matA\not=\matB}, takes on a surprisingly more difficult flavor.
Our motivation is regression, but the problem could be of more
general interest.
We will approach the problem through the use of matrix factorizations.
For \math{\matZ\in\R^{n\times k}}, with
\math{\matZ\transp\matZ=\matI_k}, let
$
\matA=\matH\matZ\transp+\matE,
$
where  \math{\matH\in\R^{m\times k}};
and, \math{\matE\in\R^{m\times n}} is the residual error.
For fixed $\matA$ and $\matZ$, \math{\XNorm{\matE}} ($\xi=2, \mathrm{F}$) is minimized when
\math{\matH=\matA\matZ}. Let \math{\matOmega
\in\R^{n\times r}}, $\matS \in \R^{r \times r}$, and \math{\matC=\matA\matOmega
\matS\in\R^{m\times r}}.

\begin{lemma} 
 \label{result1}
If $\rank(\matZ\transp \matOmega
 \matS) = k$, then,
$$
\XNorm{\matB - \matC \pinv{\matC}\matB} \le
\XNorm{\matB - \matH \pinv{\matH} \matB} +
\XNorm{\matE \matOmega
 (\matZ\transp\matOmega
)^+\pinv{\matH}\matB}.
$$
\end{lemma}
\begin{proof}
$ \XNorm{\matB - \matC \pinv{\matC} \matB}$
\begin{align*}
&\le
\XNorm{\matB - \matC \pinv{(\matZ\transp \matOmega
\matS)} \pinv{\matH} \matB}
\tag{a}\\
&=
\XNorm{\matB - \matA\matOmega\matS
 \pinv{(\matZ\transp \matOmega
\matS)} \pinv{\matH} \matB}
\\
&=
\XNorm{\matB - (\matH \matZ\transp + \matE)\matOmega
 \matS\pinv{(\matZ\transp \matOmega
\matS)}
\pinv{\matH} \matB}\\
&=
\XNorm{
\matB - \matH (\matZ\transp \matOmega
\matS) \pinv{(\matZ\transp \matOmega
\matS)} \pinv{\matH} \matB +
\matE\matOmega
 \pinv{(\matZ\transp \matOmega
\matS)} \pinv{\matH} \matB}\\
&=
\XNorm{\matB - \matH \pinv{\matH} \matB + \matE \matOmega
\matS\pinv{(\matZ\transp \matOmega
\matS)} \pinv{\matH} \matB}\tag{b}\\
&\le
\XNorm{\matB - \matH \pinv{\matH} \matB} + \XNorm{ \matE\matOmega
\matS \pinv{(\matZ\transp \matOmega
\matS)} \pinv{\matH} \matB}.\tag{c}
\end{align*}
(a) follows by the optimality of \math{\pinv{\matC} \matB}; (b) follows because
\math{\rank(\matZ\transp
\matOmega
\matS)=k} and so \math{\matZ\transp
\matOmega
\matS \pinv{(\matZ\transp
\matOmega
\matS)}=\matI_{k}}; (c) follows by the triangle inequality of matrix norms.
\end{proof}
Lemma~\ref{result1} is a general tool for the general
matrix approximation problem.
The bound has two terms which
highlight some  trade offs: the first term is the approximation
of \math{\matB}
using \math{\matH} (\math{\matH} is used in the factorization to
approximate \math{\matA}); the second term
is related to~\math{\matE}, the residual
error in approximating \math{\matA}. Ideally, one should
choose \math{\matH} and \math{\matZ} to
simultaneously
approximate \math{\matB} with \math{\matH} and have small
 residual error~\math{\matE}.
In general, these are two competing goals, and a balance
should be struck. Here, we focus on the
Frobenius norm, and will consider only
one extreme of this trade off, namely choosing the factorization to
minimize~\math{\norm{\matE}_{\mathrm{F}}
}. Specifically,
since \math{\matZ} has rank \math{k}, the best choice for
\math{\matH\matZ\transp} which minimizes \math{\FNorm{\matE}} is
\math{\matA_k}. In this case, \math{\matE=\matA-\matA_k}. Via the
SVD, \math{\matA_k=\matU_k\matSig
_k\matV_k\transp},
and so \math{\matA=(\matU_k\matSig
_k)\matV_k\transp+\matA-\matA_k}. We apply
Lemma~\ref{result1}, with $\matB = \b$, \math{\matH=\matU_k\matSig
_k},
\math{\matZ=\matV_k} and \math{\matE=\matA-\matA_k}, obtaining the
next corollary.
\begin{corollary} \label{result2x}
If $\rank(\matV_k\transp\matOmega
 \matS) = k$, then,
$$ \TNorm{\b - \matC \pinv{\matC}\b} \le
  \TNorm{\b - \matU_k \matU_k\transp \b} +
\TNorm{\matE \matOmega
\matS \pinv{(\matV_k\transp\matOmega
\matS)}\matSig
_k^{-1}\matU_k\transp\b}.$$
\end{corollary}
\noindent
Setting
$\TNorm{\b - \matC \pinv{\matC}\b} = \TNorm{\b-\matA \hat\x_r},$
with $\hat\x_r=\matOmega\matS\pinv{\matC}\b$ and
$\TNorm{\b - \matU_k \matU_k\transp \b}=
\norm{\b-\matA \x_k^*}$, we get Lemma~\ref{result2}.

\section{Related work}
A bound can be 
obtained
using the Rank-Revealing QR (RRQR) factorization~\cite{CH92a} which only 
applies to \math{r=k}:
a QR-like 
decomposition is used to select exactly $k$ columns of $\matA$
to obtain a 
sparse solution $\hat\x_k$.
Combining
Eqn.~(12) of~\cite{CH92a} with Strong RRQR~\cite{GE96} 
one gets a bound
$ \TNorm{\x^{*}_k - \hat\x_k}  \le \sqrt{4 k(n-k)+1} / \sigma_k(\matA) \cdot
\left(
2 \TNorm{\b} + \TNorm{\b - \matA \x^{*}_k}
\right).
$
We compare $\TNorm{\matA \hat\x_r - \b}$ and
$\TNorm{\matA \x^{*}_k - \b}$
and our bound is generally stronger and 
applies to any user specified $r > k$.

\paragraph{Sparse Approximation Literature}
The problem studied in this paper is  NP-hard~\cite{Nat95}.
Sparse approximation has important applications and many approximation algorithms have been proposed. 
The proposed algorithms are typically either greedy or are based on convex optimization relaxations of
the objective. We refer the reader to~\cite{Tropp1,Tropp2,Tropp3} and references therein for more details.
In general, these results try to reconstruct \math{\b} to within 
an error using the sparsest possible solution \math{\x}.
In our setting, we fix the sparsity \math{r} as a constraint and compare our
solution \math{\hat\x_r} with the 
benchmark \math{\x_k^*}.


\section{Numerical illustration}

We implemented our algorithm in Matlab and tested it on a sparse approximation problem $\TNorm{\matA\x- \b},$ where $\matA$ and $\b$
are $m \times n$ and $m \times 1$, respectively, with $m=2000$ and $n=1000$. Each element of $\matA$ and $\b$ are i.i.d. Gaussian
random variables with zero mean and unit variance. We chose $k=20$ and experimented with different values of $r = 20,30,40,...,200$.
Figure 1 shows the additive error $\TNorm{\matA \hat\x_r-\b} - \TNorm{ \matA \x_k^*-\b}$.
This experiment illustrates that the proposed algorithm computes a sparse solution vector with small approximation error.
In this case, $\TNorm{\b} \approx 25$ and $\FNorm{\matA-\matA_k}/\sigma_k(\matA) \approx 18$, so the algorithm
performs empirically better than what the worst-case bound of our main theorem predicts.

\begin{figure}[tb]
\begin{center}
\includegraphics[width=1\textwidth]{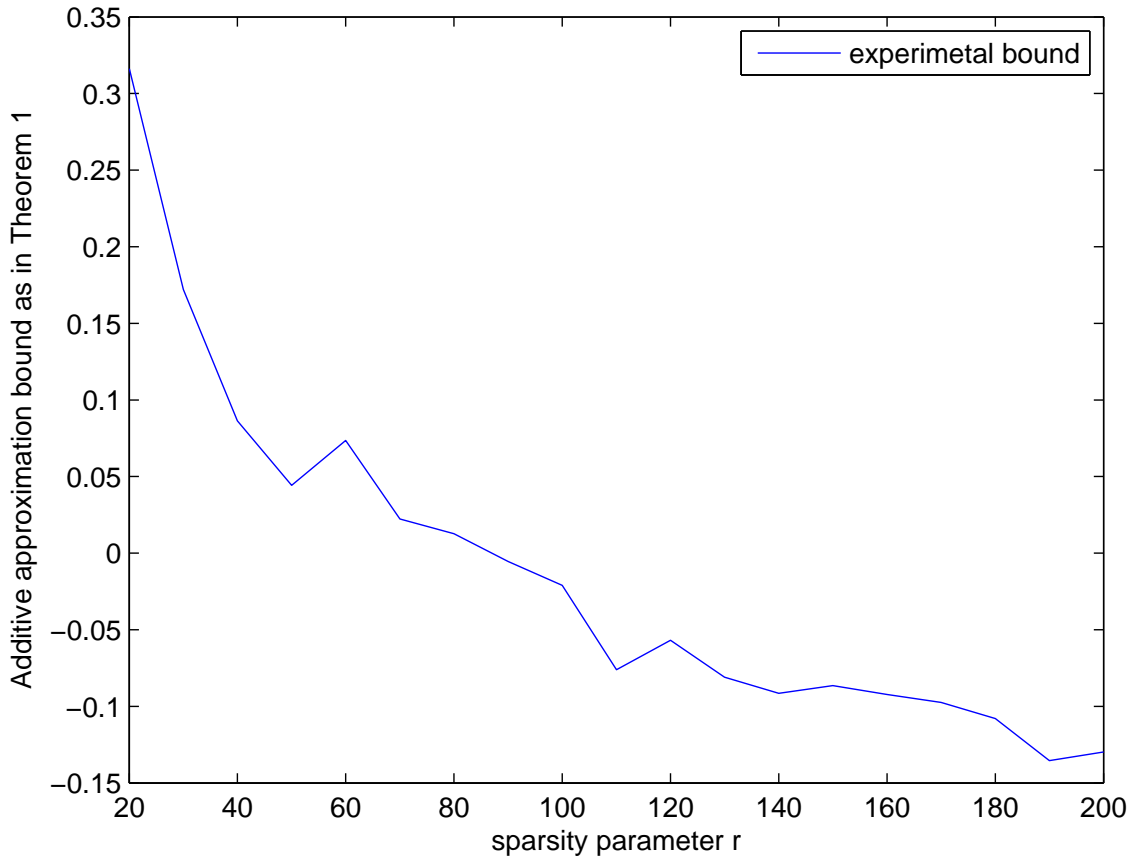}
\caption{Residual error on a problem with a $2000 \times 1000$ matrix~$\matA$.
The non-monotonic decrease arises because the algorithm chooses
columns given \math{r}, which means that the columns chosen for a 
small~\math{r} are not necessarily a subset of the columns chosen for a 
larger~\math{r}.
\label{fig:basis4}}
\end{center}
\end{figure}

\section{Concluding Remarks} 

We observe that our bound involves \math{\norm{\matA-\matA_k}_{\mathrm{F}}}. This
can be converted to a bound in terms of
\math{\norm{\matA-\matA_k}_2} using
 \math{\norm{\matA-\matA_k}_{\mathrm{F}}\le
\sqrt{n-k}\cdot\norm{\matA-\matA_k}_2}. The better bound  \math{\norm{\matA-\matA_k}_{\mathrm{F}} \le O(1+\varepsilon\sqrt{n/k})\cdot\norm{\matA-\matA_k}_2} can be obtained
by using a more expensive variant of Deterministic sampling in \cite{BDM11a} that 
bounds the spectral norm of the sampled~\math{\matE}: 
\math{\norm{\matE\matOmega\matS}_2\le (1+\sqrt{n/r})\norm{\matE}_2}. 

Sparsity in our algorithm is enforced in an unsupervised way:
the columns \math{\matC} are selected obliviously to 
\math{\b}. An interesting open question is whether the use of 
different factorizations in Lemma~\ref{result2x}, together with 
choosing the columns \math{\matC} in a \math{\b}-dependent
way
can give an error bound in terms of the optimal error
\math{\norm{\b-\matA\pinv{\matA}\b}_2}?

\section*{Acknowledgements}
Christos Boutsidis acknowledges the support from XDATA program of the Defense Advanced Research Projects Agency (DARPA), administered through Air Force Research Laboratory contract FA8750-12-C-0323. 
Malik Magdon-Ismail was partially supported 
by the Army Research Laboratory's NS-CTA program under
Cooperative Agreement Number W911NF-09-2-0053 and an NSF
CDI grant NSF-IIS 1124827.






\bibliography{RPI_BIB}







\end{document}